\documentclass{theoretics}

\title{Simple Sublinear Algorithms for \texorpdfstring{$(\Delta+1)$}{(Δ+1)} Vertex Coloring via Asymmetric Palette Sparsification}
\ThCSauthor[aff1]{Sepehr Assadi}{sepehr@assadi.info}[0009-0006-8914-5995]
\ThCSauthor[aff1]{Helia Yazdanyar}{hyazdanyar@uwaterloo.ca}[0009-0007-7194-1878]
\ThCSaffil[aff1]{School of Computer Science, University of Waterloo}
\ThCSthanks{This article was invited from SOSA 2025.~\cite{DBLP:conf/sosa/AssadiY25}\\
Part of this work was conducted while the first named author was visiting the Simons Institute for the Theory of Computing as part of the Sublinear Algorithms program. }
\ThCSshortnames{S.~Assadi, H.~Yazdanyar}
\ThCSshorttitle{Asymmetric Palette Sparsification}
\ThCSyear{2026}
\ThCSarticlenum{3}
\ThCSreceived{Feb 26, 2025}
\ThCSrevised{Aug 4, 2025}
\ThCSaccepted{Sep 1, 2025}
\ThCSpublished{Jan 29, 2026}
\ThCSdoicreatedtrue
\ThCSkeywords{Sublinear algorithms, \texorpdfstring{$(\Delta+1)$}{(Δ+1)}-coloring, vertex coloring, massively parallel computation, sublinear time, semi-streaming}

\usepackage{tcolorbox}
\tcbuselibrary{skins,breakable}
\tcbset{enhanced jigsaw}

\usepackage[normalem]{ulem}

\definecolor{DarkRed}{rgb}{0.5,0.1,0.1}
\definecolor{DarkBlue}{rgb}{0.1,0.1,0.5}

\usepackage{nameref}
\definecolor{ForestGreen}{rgb}{0.1333,0.5451,0.1333}
\definecolor{Red}{rgb}{0.9,0,0}
\usepackage{bm}

\usepackage{tikz}
\usetikzlibrary{arrows}
\usetikzlibrary{arrows.meta}
\usetikzlibrary{shapes}
\usetikzlibrary{backgrounds}
\usetikzlibrary{positioning}
\usetikzlibrary{decorations.markings}
\usetikzlibrary{decorations.pathreplacing} % needed for drawing braces
\usetikzlibrary{patterns}
\usetikzlibrary{calc}
\usetikzlibrary{fit}
\usetikzlibrary{decorations}

\usepackage[framemethod=TikZ]{mdframed}

\usepackage[capitalise,noabbrev,nameinlink]{cleveref}

\crefname{observation}{Observation}{Observations}
\crefname{claim}{Claim}{Claims}

\allowdisplaybreaks

\renewcommand{\leq}{\leqslant}
\renewcommand{\geq}{\geqslant}

\newcommand{\Ot}{\ensuremath{\widetilde{O}}}

\newcommand{\bracket}[1]{\left[#1\right]}
\newcommand{\paren}[1]{\ensuremath{\left(#1\right)}\xspace}
\newcommand{\card}[1]{\left\vert{#1}\right\vert}

\newcommand{\IN}{\ensuremath{\mathbb{N}}}

\newcommand{\expect}[1]{\Exp\bracket{#1}}

\newcommand{\set}[1]{\ensuremath{\left\{ #1 \right\}}}
\newcommand{\poly}{\mbox{\rm poly}}
\newcommand{\polylog}{\mbox{\rm  polylog}}

\DeclareMathOperator*{\Exp}{\ensuremath{{\mathbb{E}}}}
\DeclareMathOperator*{\Prob}{\ensuremath{\textnormal{Pr}}}
\renewcommand{\Pr}{\Prob}

\newenvironment{tbox}{\begin{tcolorbox}[
		enlarge top by=5pt,
		enlarge bottom by=5pt,
		 breakable,
		 boxsep=0pt,
                  left=4pt,
                  right=4pt,
                  top=10pt,
                  arc=0pt,
                  boxrule=1pt,toprule=1pt,
                  colback=white
                  ]%%
	}
{\end{tcolorbox}}

\newcommand{\II}{\ensuremath{\mathbb{I}}}

\newcommand{\mireal}[1][]{
  \ifx\relax#1\relax%
    \II(\mione \,; \mitwo)%
  \else%
    \II(\mione \,; \mitwo\mid #1)%
  \fi
}

\begin{document}

\maketitle
\begin{abstract}
The \textbf{palette sparsification theorem (PST)} of Assadi, Chen, and Khanna (SODA 2019) states that in every graph $G$ with maximum degree $\Delta$, sampling a list of $O(\log{n})$ colors from 
$\{1,\ldots,\Delta+1\}$ for 
every vertex independently and uniformly, with high probability, allows for finding a $(\Delta+1)$ vertex coloring of $G$ by coloring each vertex only from its sampled list. PST naturally leads to a host of \textbf{sublinear algorithms} 
for $(\Delta+1)$ vertex coloring, including in semi-streaming, sublinear time, and MPC models, which are all proven to be nearly optimal, and in the case of the former two are the only known sublinear algorithms for this problem. 

While being a quite natural and simple-to-state theorem, PST suffers from two drawbacks. Firstly, all its known proofs require technical arguments that rely on sophisticated graph decompositions and probabilistic arguments. 
Secondly, finding the coloring of the graph from the sampled lists efficiently requires a considerably complicated algorithm. 

We show that a natural \emph{weakening} of PST addresses both these drawbacks while still leading to sublinear algorithms of similar quality (up to polylog factors). 
In particular, we prove an \textbf{asymmetric palette sparsification theorem (APST)} that allows for list sizes of the vertices to have different sizes and only bounds the \emph{average} size of these lists. 
The benefit of this weaker requirement is that we can now easily show the graph can be $(\Delta+1)$ colored from the sampled lists using the standard greedy coloring algorithm. This way, we can recover 
nearly-optimal bounds for $(\Delta+1)$ vertex coloring in all the aforementioned models using algorithms that are much simpler to implement and analyze. 
\end{abstract}

\section{Introduction}\label{sec:intro}

Let $G=(V,E)$ be any $n$-vertex graph with maximum degree $\Delta$. A basic graph theory fact is that vertices of $G$ can be colored with $\Delta+1$ colors such that no edge is monochromatic. This can be done using 
a textbook greedy algorithm: iterate over vertices of $G$ in any arbitrary order and for each vertex greedily find a color that has not been used in its neighborhood, which is guaranteed to exist by the pigeonhole principle. 
This algorithm is quite simple and efficient as it runs in linear time and space. But, can we design even more efficient algorithms? 

Assadi, Chen, and Khanna~\cite{AssadiCK19} addressed this question by designing various \emph{sublinear} algorithms for the $(\Delta+1)$ vertex coloring problem. These are algorithms whose resource requirements 
are significantly smaller than their input size. Some canonical examples include $(a)$ \emph{graph streaming} algorithms~\cite{FeigenbaumKMSZ04} that process graphs by making one pass 
over their edges and using a limited space; $(b)$ \emph{sublinear time} algorithms~\cite{ChazelleRT01} that process graphs by 
querying its adjacency list/matrix and using limited \emph{time}; or $(c)$ \emph{Massively Parallel Computation (MPC)} algorithms~\cite{KarloffSV10} that process graphs 
in synchronous rounds and a distributed manner over machines with limited \emph{communication} (see~\Cref{sec:sublinear} for definitions). 

The key contribution of~\cite{AssadiCK19} was proving the following (purely combinatorial) theorem about $(\Delta+1)$ coloring, termed the \textbf{palette sparsification theorem (PST)}. 

\begin{theorem}[\textbf{Palette Sparsification Theorem (PST)}~\cite{AssadiCK19}]\label{thm:pst}
	For any graph $G=(V,E)$ with maximum degree $\Delta$, if we sample $O(\log{n})$ colors $L(v)$ for each vertex $v \in V$ independently and uniformly at random
	from the  colors $\set{1,2,\ldots,\Delta+1}$, then, with high probability, $G$ can be colored so that every vertex $v$ chooses its color from its own list $L(v)$. 
\end{theorem}

The authors in \cite{AssadiCK19} used PST to design sublinear algorithms in all the aforementioned models using the following ``sparsifying'' nature of this result: Firstly, it is easy to see that to color $G$ from the sampled 
lists, we can ignore all edges $(u,v) \in E$ where $L(u) \cap L(v) = \emptyset$ since they do not create any conflict under any possible coloring. Secondly, for any edge $(u,v) \in E$, the probability that it still remains conflicting---the probability that two independently chosen $O(\log{n})$-lists from $(\Delta+1)$ colors intersect---is $O(\log^2\!{(n)}/\Delta)$. Combining these with the upper bound of $n\Delta/2$ 
on the number of edges in $G$ implies that in expectation (and even with high probability) there are only $O(n\log^2\!{(n)})$  ``important'' edges that the algorithm should consider. 
The sublinear algorithms can now 
be obtained from this in a simple way\footnote{For instance, for a streaming algorithm, first sample all the $O(n\log{n})$ colors, and during the stream whenever an edge arrives, see if it is an important edge and if so store it. 
At the end, use~\Cref{thm:pst} to color the graph using these stored edges. This leads to a single-pass streaming algorithm with $O(n\log^3\!{(n)})$ bits of space for $(\Delta+1)$ coloring.}.

Since its introduction in~\cite{AssadiCK19}, PST and its variants and generalizations have been studied in sublinear algorithms~\cite{ChangFGUZ19,BeraCG20,AssadiKM22,ChakrabartiGS22,AssadiCGS23}, distributed computing~\cite{FischerHM23,FlinGHKN23,HalldorssonKNT22}, and discrete mathematics~\cite{AlonA20,AndersonBD22,KahnK23,KahnK24,HefetzK24,AshvinkumarK24}. 

Despite its long list of applications, and perhaps even due to its generality, PST suffers from two drawbacks. Firstly, the proof of PST is considerably complicated and technical; it goes through 
the so-called \emph{sparse-dense decomposition} of graphs due to~\cite{Reed98} (and the variant proved in~\cite{AssadiCK19} itself, building on~\cite{HarrisSS16}), and then a detailed three-phase approach
for coloring each part in the decomposition differently using various probabilistic and random graph theory arguments\footnote{All known proofs of PST for $(\Delta+1)$ coloring (or so-called $(\deg+1)$-(list) 
coloring)~\cite{AssadiCK19,AlonA20,HalldorssonKNT22,KahnK23,FlinGHKN24,AshvinkumarK24} follow this decomposition plus three-phase approach and it was even pointed out in~\cite{KahnK23} that: ``something of this type [...] seems more or less
unavoidable''.}. Secondly, PST, as stated, is a purely combinatorial theorem and not an algorithmic one, and thus
to obtain the coloring of $G$, one needs to run a rather complicated and non-greedy post-processing algorithm to find the coloring of $G$ from the sampled lists (given also the decomposition of the graph; see~\cite{AssadiCK19} for more details). 

\paragraph{Our contribution.} We show that introducing a simple \emph{asymmetry} in the definition of PST leads to the same colorability guarantee but this time using a much simpler proof and algorithm. 

\begin{quote} %\begin{result}
\textbf{Asymmetric Palette Sparsification Theorem (APST).}
	For any graph $G=(V,E)$ with maximum degree $\Delta$, there is a distribution on list-sizes $\ell: V \rightarrow \IN$ (depending only on vertices $V$ and \underline{not} edges $E$) 
	such that an \underline{average} list size is $O(\log^2\!{(n)})$ and the following holds. With high probability, 
	if we sample $\ell(v)$ colors $L(v)$ for each vertex $v \in V$ independently and uniformly at random
	from colors $\set{1,2,\ldots,\Delta+1}$, then, with high probability, the greedy coloring algorithm that processes vertices in some fixed order (determined by vertex degrees and list sizes\footnote{For $\Delta$-regular
	graphs, this translates to processing vertices in the increasing order of their list sizes.}) 
	finds a proper coloring of $G$ by coloring each $v$ from its own list $L(v)$. 
\end{quote} %\end{result}

The benefit of our APST compared to the original PST of~\cite{AssadiCK19} (and all its other alternative variants in~\cite{AlonA20,HalldorssonKNT22,KahnK23,FlinGHKN24,AshvinkumarK24}) is twofold: it admits a 
\emph{significantly} simpler proof, while also allowing for coloring from the sampled list using the standard greedy algorithm itself. At the same time, it is also weaker than the original PST 
on two fronts: it requires $O(\log^2\!{(n)})$ list sizes per vertex as opposed to $O(\log{n})$ (which is similar to~\cite{HalldorssonKNT22,FlinGHKN24}) but much more importantly, it allows the vertices to have much larger 
list sizes in the worst-case and only bounds the \emph{average} list sizes (this is different from all prior work on PST, and was inspired by the recent work of~\cite{FlinM24} on the communication complexity of $(\Delta+1)$ coloring). We note that asymmetric list sizes are \emph{necessary} if we would like to color the graph \emph{greedily} from the sampled colors\footnote{Consider 
coloring a $(\Delta+1)$-clique. When coloring the last vertex $v$ of the clique in the greedy algorithm, there is only one color that can be assigned to $v$ but to ensure this color is sampled by $v$ even with a constant probability, we need $L(v)$ to have 
size $\Omega(\Delta)$. Our APST addresses this by allowing vertices that are colored later in the greedy algorithm to also sample more colors, while earlier vertices can stick with smaller list sizes.}. 

Finally, APST allows for recovering the same type of sublinear algorithms as in~\cite{AssadiCK19} with only an extra $\poly\!\log\!{(n)}$ overhead in their resources. 
In particular, we obtain the following sublinear algorithms for $(\Delta+1)$ vertex coloring: 
\begin{itemize}
	\item A randomized \textbf{graph streaming} algorithm with $\Ot(n)$ space\footnote{Throughout, we use $\Ot(f) := O(f \cdot \poly\log{(f)})$ to suppress polylog factors.} and a single pass over the input.  
	\item A randomized \textbf{sublinear time} algorithm with $\Ot(n^{1.5})$ time, given (non-adaptive) query access to both adjacency list and matrix of the input (also called the general query model). 
	\item A randomized \textbf{MPC} algorithm with $\Ot(n)$ memory per machine and $O(1)$ rounds. 
\end{itemize}
The number of passes in our streaming algorithm is clearly optimal. The number of rounds in the MPC algorithm is
also asymptotically optimal (in fact, it is only two rounds, and even one round assuming access to public randomness). It was also proven in~\cite{AssadiCK19} that 
the runtime of the sublinear time algorithm and the space of the streaming algorithm are nearly optimal up to $\poly\!\log{(n)}$ factors (for the streaming algorithm, storing the coloring itself requires this much space anyway). 
These constitute the simplest known sublinear algorithms for $(\Delta+1)$ vertex coloring.\footnote{To our knowledge, no other streaming nor sublinear time algorithms beside~\cite{AssadiCK19} have been developed for this problem but for MPC algorithms,~\cite{ChangFGUZ19,CzumajDP21} have subsequently presented other algorithms for this problem.} 

In conclusion, we find our APST to be a more ``algorithm friendly'' version of PST that still allows for recovering nearly optimal sublinear algorithms for $(\Delta+1)$ coloring. In addition, we hope that its proof can act as a gentle warm up and introduction to the original PST itself. 

Finally, in~\Cref{sec:warm-up}, we also present an even simpler sublinear \emph{time} algorithm for $(\Delta+1)$ coloring whose proof
is inspired by our APST but does not directly use it. While qualitatively weaker, the benefit of this algorithm is that it is entirely self-contained and only requires elementary probabilistic arguments (not even concentration inequalities), and 
thus can be even more ``classroom friendly'' than the algorithms obtained from our APST.

%------------------------------------------ Perliminaries-------------------------------------------------
% \input{prelim}

\newcommand{\degback}[1]{\ensuremath{\deg^{<}_{\pi}(#1)}}
\newcommand{\Nback}[1]{\ensuremath{N^{<}_{\pi}(#1)}}

\section{Preliminaries}\label{sec:prelim}

\paragraph{Notation.} For any integer $t \geq 1$, define $[t] := \set{1,2,\ldots,t}$. For a graph $G=(V,E)$, we use $n$ to denote the number of vertices and for each $v \in V$, use $N(v)$ to denote the neighbors of $v$ and $\deg(v)$ as its degree. 

We say a probabilistic event happens \emph{with high probability}, if it happens with probability at least $1-1/\poly{(n)}$, where $n$ is the number of vertices in the graph which will be clear from the context. 
We also say a probabilistic event happens \emph{with certainty}, if it happens with probability $1$. 

\paragraph{Concentration inequalities.} 

A \textbf{hypergeometric} random variable with parameters $N$, $K$, and $M$ is a discrete random variable in $\IN$ distributed as follows: we have $N$ elements, $K$ of them are marked as `good', 
and we sample $M$ elements uniformly at random and \emph{without} replacement and count the number of good samples. We use a standard result on the concentration of hypergeometric 
random variables. 

\begin{proposition}[cf.~{\cite[Theorem 2.10]{JansonLR11}}]\label{thm:book}
Suppose $X$ is a hypergeometric random variable with parameters $N,K,M$ and thus the expectation $\expect{X} = M \cdot K/N$. Then, for any $t \geq 0$
\[
\Pr(X \leq \expect{X}-t) \leq \exp\paren{-\frac{t^2}{2\expect{X}}}.
\]
\end{proposition}

\subsection{Sublinear Models Considered in this Paper}\label{sec:sublinear}

For completeness, we present a brief definition of each of the computational models that we consider. %We refer the interested reader to~\cite{AssadiCK19} for more background on these models. 

\paragraph{Graph streaming.} In this model, the input graph $G=(V,E)$ is presented to the algorithm as a stream of its edges ordered in some arbitrary manner. 
The algorithm makes a single pass (or sometimes multiple ones) over this stream while using $\Ot(n)$ memory and at the end of the stream, outputs a solution to the problem on the input graph $G$.  
See~\cite{FeigenbaumKMSZ04} for more details. 

\paragraph{Sublinear time.} In this model, the input graph $G=(V,E)$ is presented to the algorithm via query access to its adjacency list and matrix or alternatively speaking, using degree-, neighbor-, and pair-queries. 
The algorithm can make its queries adaptively or non-adaptively and then at the end outputs a solution to the problem on the graph $G$. See~\cite{ChazelleRT01} for more details. 

\paragraph{MPC.} In this model, the input graph $G=(V,E)$ is originally partitioned across multiple machines, each with $\Ot(n)$ memory. Computation happens in synchronous rounds wherein each machine can send and receive 
$\Ot(n)$-size messages. After the last round, one designated machine outputs a solution to the problem on the graph $G$. See~\cite{KarloffSV10} for more details. 

%------------------------------------------ APST-------------------------------------------------
% \input{apst}

\newcommand{\Nforward}[1]{\ensuremath{N^{>}_{\pi}(#1)}}
\newcommand{\degForward}[1]{\ensuremath{\deg^{>}_{\pi}(#1)}}

\section{Asymmetric Palette Sparsification}\label{sec:apst}

We present our APST in a more general form that allows for coloring vertices from arbitrary palettes of size proportional to their individual degrees---instead of the same original palette of $\set{1,2,\ldots,\Delta+1}$ described earlier--- namely, the $(\deg+1)$-list coloring problem. We then obtain the APST for $(\Delta+1)$ coloring as a simple corollary of this result. Furthermore, we remark that a PST version of $(\deg+1)$-list coloring (by sampling $\polylog{(n)}$ colors per vertex in the worst-case) 
has been previously established in~\cite{HalldorssonKNT22}; see also~\cite{AlonA20}. 

%% Generalization deg+1 APST Theorem--------------------
\begin{theorem}[\textbf{Asymmetric Palette Sparsification Theorem (for ($\deg+1$)-list coloring)}]\label{thm:apst-deg+1}
	Let $G=(V,E)$ be any $n$-vertex graph together with lists $S(v)$ with $\deg(v)+1$ arbitrary colors for each vertex $v\in V$.
	Sample a random permutation $\pi: V \rightarrow [n]$ uniformly and define 
	\[
	\ell(v) := \min \left( \deg(v)+1\, , \,  \frac{40\,n\ln{n}}{\pi(v)}\right),
	\]
	for every $v \in V$ as the size of the list of colors to be sampled for vertex $v$. Then,  
	\begin{itemize}
		\item \emph{\textbf{List sizes:}} $\sum_{v \in V} \ell(v) = O(n\log^2\!{(n)})$ with certainty; and, for any fixed vertices $u \neq v \in V$, 
		\[
		\expect{\ell(v)} = O(\log^2\!{(n)}) \quad \text{and} \quad \expect{\ell(u) \cdot \ell(v)} = O(\log^{4}\!{(n)}). 
		\] 
		\item \emph{\textbf{Colorability:}}
		If for every vertex $v \in V$ we sample a list $L(v)$ of $\ell(v)$ colors from $S(v)$ uniformly and independently, then, with high probability (over the randomness of $\ell$ and sampled lists) the greedy algorithm that iterates over vertices $v \in V$ in the decreasing order of $\pi(v)$ finds a proper list-coloring of $G$ from the lists 
		$\set{L(v) \mid v \in V}$. 
	\end{itemize}
	
\end{theorem}
	\Cref{thm:apst-deg+1} follows immediately from~\Cref{clm:apst-list-sizes} and~\Cref{clm:apst-colorability} proven below. 
	\begin{lemma}[\textbf{List sizes}]\label{clm:apst-list-sizes}
		We have $\sum_{v \in V} \ell(v) = O(n\log^2\!{(n)})$ with certainty; and, for $u \neq v \in V$, 
		\[
		\expect{\ell(v)} = O(\log^2\!{(n)}) \qquad \text{and} \qquad \expect{\ell(u) \cdot \ell(v)} =O(\log^{4}\!{(n)}). 
		\]
	\end{lemma}
	\begin{proof}
		By the definition of $\ell(v)$, 
		\[
		\sum_{v \in V} \ell(v) \leq \sum_{v \in V} \frac{40\,n\ln{n}}{\pi(v)} = \sum_{i=1}^{n} \frac{40\,n\ln{n}}{i} = O(n\log^2{n}),
		\]
		by the sum of the harmonic series. This proves the deterministic bound on the total size of the lists as well as the first inequality for the expected size of each list, 
		given the distribution of list sizes is symmetric across vertices. For the second inequality, we have
        \begin{align*}
            \mathbb{E}\bracket{\ell(u) \cdot \ell(v)}  
            &\leq \sum_{i \neq j} \textnormal{Pr}(\pi(v)=i \, \wedge \, \pi(u)=j) \cdot \frac{(40\,n\ln{(n)})^2}{i\cdot j} \\
            &= \frac{1600\,n^2\ln^2{(n)}}{{n\cdot (n-1)}} \cdot \sum_{i =1}^{n} \sum_{j=1}^{i-1} \frac{1}{i \cdot j}\\
            &= O(\log^4(n)),
        \end{align*}
		% \[
		% \mathbb{E}\bracket{\ell(u) \cdot \ell(v)}  \leq \sum_{i \neq j} \textnormal{Pr}(\pi(v)=i \, \wedge \, \pi(u)=j) \cdot \frac{(40\,n\ln{(n)})^2}{i\cdot j} = \frac{1600\,n^2\ln^2{(n)}}{{n\cdot (n-1)}} \cdot \sum_{i =1}^{n} \sum_{j=1}^{i-1} \frac{1}{i \cdot j} = O(\log^4(n)),
		% \]
		again, by the sum of the harmonic series. This concludes the proof. 
	\end{proof}
	
	Establishing the colorability property is the main part of the proof. We first need some notation. 
%%	Iterating over vertices in the increasing order of $\ell(\cdot)$ is the same as the decreasing order of $\pi(\cdot)$. 
	For any vertex $v \in V$, define $\Nback{v}$ as the neighbors $u$ of $v$ with $\pi(u) < \pi(v)$, namely the ones that are processed 
	\emph{after} $v$ by the greedy algorithm. Let  $\degback{v}:= \card{\Nback{v}}$. 
	
	An easy observation is that when coloring any vertex $v$ by the greedy algorithm (using the palette of $\Delta+1$ colors), there are still at least $\degback{v}+1$ colors not used in the neighborhood of $v$. 
	Thus, lower bounding $\degback{v}$ allows us to later prove that $L(v)$ contains an available color for $v$ with high probability.  
	We only need such a bound for the subset of vertices with $\ell(v) < \deg(v)+1$ as captured in the following claim. 
	
	\begin{claim}\label{clm:apst-a(v)}
		For all $v \in V$, 
		\[
		\Pr\paren{\degback{v} < \frac{\deg(v) \cdot \pi(v)}{2n} \,\middle\vert\, \pi(v) > \frac{40n\ln(n)}{\deg(v)+1}} \leq n^{-2.4}.
		\]
	\end{claim}
	\begin{proof}
		Fix any vertex $v$ as above and condition on $\pi(v) = i+1$ for some $i \in \set{0,1,\ldots,n-1}$. 
		For $j \in [\deg(v)]$, define $X_j \in \set{0,1}$ as the 
		indicator random variable which is $1$ iff the $j$-th neighbor of $v$ appears in the first $i$ vertices of $\pi$, and is therefore processed after $v$. 
		For $X$ being the sum of these $X_j$-variables (and implicitly conditioned on $\pi(v) = i+1$ to avoid cluttering the equations), 
		\[
		\expect{\degback{v} \mid \pi(v)=i+1} = \expect{X} = \sum_{j=1}^{\deg(v)} \expect{X_j} = \deg(v) \cdot \frac{i}{n-1},
		\]
		as each neighbor of $v$ appears in the first $i$ position of $\pi$ with probability $i/(n-1)$ conditioned on $\pi(v) = i+1$.  
		Random variable $X$ is distributed as a hypergeometric random variable with parameters $N=n-1$, $K = i$, and $M = \deg(v)$. Thus, by~\Cref{thm:book} for the parameter
		\begin{align}
		t &:=  \expect{X} - \frac{\deg(v)\cdot (i+1)}{2n} = \frac{\deg(v) \cdot i}{n-1} - \frac{\deg(v)\cdot (i+1)}{2n} \geq  (1-o(1)) \cdot \frac{\expect{X}}{2}  \tag{by the value of $\expect{X}$ calculated above and given $i = \omega(1)$ in the claim statement}
		\end{align} 
		conditioning on $\pi(v) > 40n\ln(n)/(\deg(v)+1)$ we have, 
		\begin{align*}
			\Pr\paren{X \leq \frac{\deg(v)\cdot (i+1)}{2n}} &= \Pr\paren{X \leq \expect{X} - t} \leq \exp\paren{-\frac{t^2}{2\expect{X}}} \leq \exp\paren{-(1-o(1)) \cdot \frac{\expect{X}}{8}} \tag{by the lower bound on $t$ established above} \\
			&= \exp\paren{-(1-o(1)) \cdot \frac{\deg(v) \cdot i}{8 \cdot (n-1)}} \leq \exp\paren{-(1-o(1)) \cdot \frac{40\ln{n}}{8}} \leq n^{-2.4},
		\end{align*}
		where in the second line we used the lower bound on $i = \pi(v)-1$ in the condition. 
		Since $\degback{v} = X$ and the bound holds for all choices of $\pi(v)$, we can conclude the proof.
	\end{proof}
	\begin{lemma}[\textbf{Colorability}]\label{clm:apst-colorability}
		With high probability, when coloring each vertex $v \in V$ in the greedy algorithm, 
		at least one of the colors sampled in $L(v)$ has not been used in the neighborhood of $v$; that is, the greedy algorithm can color this vertex. 
	\end{lemma}
	\begin{proof}	
		We condition on the choice of $\pi$ and by union bound obtain that with high probability, the complement of the event in~\Cref{clm:apst-a(v)} holds for all vertices. 
		For any vertex $v \in V$, we say a color in $S(v)$ is \textbf{available} to $v$ iff it is not assigned to any neighbor of $v$ by the time we process $v$ in the greedy algorithm. 
		Let $a(v)$ denote the number of available colors and note that $a(v) \geq  \degback{v}+1$. 
		
		In the following, we fix a vertex $v \in V$ and further condition on the randomness of $L(u)$ for all vertices $u$ with $\pi(u) > \pi(v)$. As we are using a fully deterministic greedy algorithm, the conditioning so far fix the set of available colors and the values of $a(v)$ and $\ell(v)$, but $L(v)$ is still a random $\ell(v)$-subset of the $\deg(v)+1$ colors in $S(v)$. If $\pi(v) \leq 40n\ln(n)/(\deg(v)+1)$, we have $\ell(v) = \deg(v)+1$ which means $L(v) = S(v)$ and thus there exists an available color in $L(v)$, proving the claim for this vertex. 
		
		For the remaining vertices with $\pi(v) > 40n\ln(n)/(\deg(v)+1)$, we can apply~\Cref{clm:apst-a(v)} to have, 
		\begin{align*}
			\textnormal{Pr}\left(\text{no available color of $v$ is sampled in $L(v)$}\right) &\leq (1-\frac{a(v)}{\deg(v)+1})^{\ell(v)} \leq \exp\left(-\frac{a(v) \cdot \ell(v)}{\deg(v)+1}\right) \\ 
			& \leq \exp\left(-\frac{\deg(v) \cdot \pi(v)}{2n} \cdot \frac{40n\ln{n}}{\pi(v)} \cdot \frac{1}{\deg(v)+1}\right) \\
            &\leq n^{-5}.
		\end{align*}
		Thus, with high probability, we can color $v$ from $L(v)$ in the greedy algorithm. Taking a union bound over all vertices concludes the proof. 
	\end{proof}

We obtain our APST for $(\Delta+1)$ coloring described earlier as a direct corollary of~\Cref{thm:apst-deg+1}. 

\begin{corollary}[\textbf{Asymmetric Palette Sparsification Theorem (for ($\Delta+1$) coloring)}]\label{thm:apst} ~ \\
	Let $G=(V,E)$ be any $n$-vertex graph with maximum degree $\Delta$. 
	Sample a random permutation $\pi: V \rightarrow [n]$ uniformly and define 
	\[
	\ell(v) := \min \left( \Delta+1\, , \,  \frac{40\,n\ln{n}}{\pi(v)}\right),
	\]
	for every $v \in V$ as the size of the lists of colors to be sampled for vertex $v$. Then,  
	\begin{itemize}
		\item \emph{\textbf{List sizes:}} $\sum_{v \in V} \ell(v) = O(n\log^2\!{(n)})$ with certainty; and for any fixed vertices $u \neq v \in V$, 
		\[
		\expect{\ell(v)} = O(\log^2\!{(n)}) \quad \text{and} \quad \expect{\ell(u) \cdot \ell(v)} = O(\log^{4}\!{(n)}). 
		\] 
		\item \emph{\textbf{Colorability:}} If for every vertex $v \in V$, we sample a list $L(v)$ of $\ell(v)$ colors from $[\Delta+1]$ uniformly and independently, then, with high probability (over the randomness of $\ell$ and sampled lists) the greedy 
		algorithm that iterates over vertices $v \in V$ in the decreasing order of $\pi(v)$ finds a proper list-coloring of $G$ from the lists 
		$\set{L(v) \mid v \in V}$. 
	\end{itemize}
%%	The distribution over $\ell: V \rightarrow \IN$ is as follows. 
\end{corollary}

\begin{proof}
\noindent 
	The proof of the \textbf{List Size} property is identical to the proof of~\Cref{clm:apst-list-sizes} as we only use the fact $\ell(v)\leq 40\,\ln{n}/\pi(v)$ in that proof, which continues to hold here as well. Thus, we avoid repeating the same arguments. 
	
To prove the \textbf{Colorability} property, consider the following equivalent process of sampling $L(v)$ in~\Cref{thm:apst} after having picked $\pi$: for each vertex $v \in V$, first sample $\deg(v)+1$ colors from $[\Delta+1]$ uniformly and put them in a list $S(v)$,
then pick $\ell'(v)= \min\{\deg(v)+1, \ell(v)\}$ colors from $S(v)$ uniformly and put them in a list $L'(v)$, and finally, sample $\ell(v) - \ell'(v)$ colors from $[\Delta+1] \setminus L'(v)$ and 
add them to $L'(v)$ to obtain $L(v)$. 

This way, we have that $L(v)$ is still a random $\ell(v)$-subset of $[\Delta+1]$ as desired by~\Cref{thm:apst}, but now, we also have that conditioned on the choice of $S(v)$, $L'(v)$ is a 
random $\ell'(v)$-subset of $S(v)$ as in~\Cref{thm:apst-deg+1}. Thus, we can apply~\Cref{thm:apst-deg+1} to $G$ and $\set{S(v) \mid v \in V}$ and argue that if we go over vertices of $G$ in the decreasing order of $\pi$, 
the greedy algorithm with high probability can color the graph from $L'(v) \subseteq L(v)$ also, concluding the proof.  
\end{proof}
%%
%%We conclude this section by remarking that, in the \textbf{Colorability} property of~\Cref{thm:apst}, iterating over the vertices (in the greedy algorithm)
%%in the decreasing order of $\pi$ is effectively the same as increasing order of list-sizes plus a certain tie-breaking dictated by $\pi$. 
%%	

%--------------------- Sublinear-------------------------------------------------
% \input{sublinear}

\newcommand{\LL}{\ensuremath{\mathcal{L}}}

\section{Sublinear Algorithms from Asymmetric Palette Sparsification}\label{sec:asymmetric} 

We now show how our asymmetric palette sparsification theorem in~\Cref{sec:apst} can be used to obtain sublinear algorithms for $(\Delta+1)$ vertex coloring. These algorithms are more or less identical 
to the (exponential-time) sublinear algorithms of~\cite{AssadiCK19} from their original palette sparsification theorem and we claim no novelty in this part\footnote{We do note that however, the time-efficient algorithms of~\cite{AssadiCK19} are considerably more complex. They first need to find a so-called sparse-dense decomposition of the input graph via sublinear algorithms. Then, this decomposition is used to color the final graph from the sampled colors using an algorithmic version of the proof of their palette sparsification theorem which in particular requires a non-greedy and not-so-simple approach.}. Instead, we merely point out how the ``asymmetry'' in list-sizes in~\Cref{thm:apst} 
does not weaken the performance of the resulting sublinear algorithms beyond some $\polylog{(n)}$-factors, but instead leads to much simpler post-processing algorithms for finding the coloring of the graph from the sampled lists. 

\begin{theorem}\label{thm:sublinear}
	There exist randomized sublinear algorithms that given any graph $G=(V,E)$ with maximum degree $\Delta$ with high probability output a $(\Delta+1)$ vertex coloring of $G$ using: 
	\begin{itemize}
		\item \emph{\textbf{Graph streaming:}} a single pass over the edges of $G$ in any order and $\Ot(n)$ space; 
		\item \emph{\textbf{Sublinear time:}} $\Ot(n^{1.5})$ time and non-adaptive queries to adjacency list and matrix of $G$;
		\item \emph{\textbf{Massively parallel computation (MPC):}} $O(1)$ rounds with machines of $\Ot(n)$ memory.  
	\end{itemize}
\end{theorem}

As stated earlier, the proof of~\Cref{thm:sublinear} follows the same exact strategy as the (exponential-time) sublinear algorithms of~\cite{AssadiCK19}. To do so, we need the following definition. 

\paragraph{Conflict graphs.} Let $G=(V,E)$ be any graph with maximum degree $\Delta$ and $\LL := \set{L(v) \mid v \in V}$ be a set of lists of colors sampled for vertices of $G$ according to the distribution of~\Cref{thm:apst}. 
Note the choice of $\LL$ only depends on $\Delta$ and vertices of $V$, but not edges $E$. 
We define the \textbf{conflict graph} $G_\LL = (V, E_{\LL})$ of $G$ and $\LL$ as the spanning subgraph of $G$ consisting of all edges $(u,v) \in E$ such that the sampled lists $L(u)$ and $L(v)$ intersect with each other. 

The following observation allows us to use conflict graphs in our sublinear algorithms as a proxy to the original graph $G$. 

\begin{observation}\label{obs:conflict-graph-use}
	The greedy algorithm in~\Cref{thm:apst} outputs the same exact coloring when run over the conflict graph $G_{\LL}$ instead of the original graph $G$. 
\end{observation}
\begin{proof}
	The only edges that affect the greedy algorithm of~\Cref{thm:apst} are edges $(u,v) \in E$ such that $L(u) \cap L(v)$ is non-empty. These edges are identical in $G$ and $G_{\LL}$. 
\end{proof}

The following easy claim also allows us to bound the size of the conflict graphs. 

\begin{claim}\label{clm:conflict-graph-size}
	The list of colors $\LL$ consists of $O(n\log^2{n})$ colors with certainty and the expected number of edges in $G_\LL=(V,E_{\LL})$ is $\Exp\card{E_{\LL}} = O(n\log^4{n})$. 
\end{claim}
\begin{proof}
	By the list sizes property of \Cref{thm:apst}, we already know that  $\LL$ consists of $O(n\log^2{n})$ colors with certainty. For the second part, for any edge $(u,v)$, we have, 
	\[
		\Pr\paren{\text{$(u,v)$ is in $E_{\LL}$}} \leq \Exp_{\ell(u),\ell(v)}\bracket{\Pr\paren{L(u) \cap L(v) \neq \emptyset \mid \ell(u),\ell(v)}} \leq \expect{\frac{\ell(u) \cdot \ell(v)}{\Delta+1}} = \frac{O(\log^4{(n)})}{\Delta},
	\]
	where the first inequality is by the law of conditional expectation (by conditioning on list-sizes first), the second is by union bound, and the third is by the list-size properties of~\Cref{thm:apst}. Since the total
	number of edges in $G$ is at most $n\Delta/2$, we can conclude the proof. 	
\end{proof}

We now prove~\Cref{thm:sublinear} for each family of sublinear algorithms separately. In the following, we prove the resource guarantees of the algorithms only in expectation instead of a worst-case bound that happens with certainty. However, using the standard trick of running $O(\log{n})$ copies of the algorithm in parallel, terminating any copy that uses more than twice the expected resources, and returning the answer of any of the remaining ones, we obtain the desired
algorithms in~\Cref{thm:sublinear} as well (this reduction only increases space/query/memory of algorithms with an $O(\log{n})$ multiplicative factor and the error probability with a $1/\poly(n)$ additive factor). 

\paragraph{Graph streaming.} At the beginning of the stream, sample the colors $\LL$ using~\Cref{thm:apst} and during the stream, only store the edges that belong to $G_{\LL}$. This is possible as we only need a random permutation and maximum degree $\Delta$ in order to provide list sizes of~\Cref{thm:apst} and sample the colors $\LL$. In the end, 
run the greedy algorithm on $G_{\LL}$ and return the coloring.~\Cref{thm:apst} ensures that with high probability $G$ is (list-)colorable from the sampled lists which leads to a $(\Delta+1)$ coloring of the entire graph. 
\Cref{obs:conflict-graph-use} ensures that we only need to work with $G_{\LL}$ at the end of the stream and not all of $G$, and~\Cref{clm:conflict-graph-size} bounds the space of the algorithm with $\Ot(n)$ space in expectation. 

We can also implement this algorithm in \emph{dynamic streams} by recovering the conflict graph using a sparse recovery sketch instead of explicitly storing each of its edges in the stream. See~\cite{AhnGM12} for more on dynamic streams. 

Finally, the knowledge of $\Delta$ is not necessary for this algorithm and can be removed using the same argument as in~\cite{AssadiCK19}. 

% change ${{n}\choose{2}}$ into $\binom{n}{2}$ due to ``Package amsmath Warning: Foreign command \atopwithdelims''
\paragraph{Sublinear time.} By the same argument as in our graph streaming algorithm, we can sample the colors $\LL$ using~\Cref{thm:apst} and query \emph{all} pairs of vertices $u \neq v \in V$ where $L(u) \cap L(v)$ is non-empty to find the edges of $G_{\LL}$. The same analysis as in~\Cref{clm:conflict-graph-size}
applied to the $\binom{n}{2}$ vertex-pairs (instead of $\leq n\Delta/2$ edges), ensures that the expected number of queries is $\Ot(n^2/\Delta)$. We can then color $G_{\LL}$ using the greedy algorithm in $\Ot(n)$ time. The correctness follows 
from~\Cref{thm:apst} and \Cref{obs:conflict-graph-use} as before. This algorithm only requires adjacency matrix access to $G$ and we run it when $\Delta \geq \sqrt{n}$ to obtain an $\Ot(n\sqrt{n})$ time/query algorithm. When $\Delta \leq \sqrt{n}$, we instead run the standard greedy algorithm using $O(n\Delta) = \Ot(n\sqrt{n})$ 
time using the adjacency list access to $G$. 

\paragraph{MPC.} The algorithm is almost identical to the semi-streaming one. Suppose we have access to public randomness. Then, we can sample the lists $\LL$ publicly, and each machine that has an 
edge in $G_{\LL}$ can send it to a designated machine. This way, a single machine receives all of $G_{\LL}$ and can color it greedily. The correctness follows from~\Cref{thm:apst} and~\Cref{obs:conflict-graph-use} and the memory needed for this designated machine will be $\Ot(n)$ in expectation by~\Cref{clm:conflict-graph-size}. 
Finally, we can remove the public randomness by having one machine do the sampling first on its own and share it with all the remaining machines in $O(1)$ rounds using the standard MPC primitives of search and sort.

%------------------------------------------ conclusion-------------------------------------------------
% \input{conclusion}
\section{Concluding Remarks}\label{sec:conc}

In this paper, we simplified the sublinear algorithms of~\cite{AssadiCK19} for $(\Delta+1)$ vertex coloring of graphs with maximum degree $\Delta$. This was achieved by \emph{weakening} 
the palette sparsification theorem (PST) of~\cite{AssadiCK19} to our \emph{asymmetric} palette sparsification theorem (APST), which we show admits a much simpler proof and coloring algorithm, without
limiting its applicability to sublinear algorithms. At this point, there are several natural questions for future work, which we elaborate on below.

Since its introduction in~\cite{AssadiCK19}, PST has found various applications in designing sublinear algorithms for other graph coloring problems as well. It will be interesting to examine to what extent 
APST is tailored to $(\Delta+1)$ coloring and whether, similar to PST, APST can also be extended to other graph coloring problems. For instance,~\cite{AlonA20} developed sublinear algorithms for
$O(\Delta/\log(\Delta))$ coloring of triangle-free graphs using a different PST targeted to this problem. Can we also design an APST for this problem? 

Another set of questions are regarding ``white-box'' applications of PST. For instance,~\cite{AssadiKM22} showed how to use ideas 
from PST, plus various other tools, to design a graph streaming algorithm for $\Delta$-coloring. Can we use APST to obtain a similar result but without the complicated algorithms and elaborate case analyzes in~\cite{AssadiKM22}? 
Taking this even further, can one obtain sublinear algorithms for $(\Delta-1)$ coloring~\cite{Reed99} or even less colors as in~\cite{MolloyR14} wherein the barrier to colorability are ``small'' subgraphs%
\footnote{These results are known in the graph theory literature as strengthening of Brook's theorem for $\Delta$-coloring which identifies having a $(\Delta+1)$-clique as the only barrier to $\Delta$ colorability for $\Delta \geq 3$. 
A highly general form of these results is that of~\cite{MolloyR14} that shows the only barriers to $\Delta-k+1$ coloring, for large enough $\Delta$, are $O(\Delta)$-size subgraphs that are not $(\Delta-k+1)$ colorable, as long
as $k$ satisfies $(k+1) \cdot (k+2) \leq \Delta$, and that this threshold is tight~\cite{EmdenHK98}.}. Another example of these white-box applications is in designing Local Computation Algorithms (LCA) for $(\Delta+1)$ coloring~\cite{ChangFGUZ19}; again, 
can APST lead to simpler and more efficient LCAs for $(\Delta+1)$ coloring, e.g., one with only $O(\Delta \cdot \log{(n)})$ query time (existing algorithms in~\cite{ChangFGUZ19} have some unspecified $\poly(\Delta)$ dependence)?
Yet another example, is designing dynamic graph algorithms for $(\Delta+1)$ coloring, against adaptive adversaries~\cite{BehnezhadRW25,FlinH25} that use sparse-dense decompositions and somewhat similar coloring steps as in the PST. Again, can ideas from APST simplify and further improve these bounds? It is worth mentioning that already some tools developed in this paper are used in~\cite{BehnezhadRW25} for addressing this problem.

Finally, beside its applicability to sublinear algorithms, PST has also been studied extensively from a purely combinatorial point of view~\cite{AlonA20,AndersonBD22,KahnK23,KahnK24,HefetzK24,AshvinkumarK24}. It thus is natural to consider APST through the same lens as well; for instance, in the same vein as in~\cite{HefetzK24}, we can ask the following question: suppose we sample lists of \emph{average} size only \emph{two} (or some other larger \emph{constant}) 
colors on vertices chosen from some range $\set{1,\ldots,q}$; then, what is a (asymptotically) minimal choice for $q$ so that the greedy algorithm that colors the vertices in the increasing order of their list sizes 
 finds a proper coloring of any given graph from its sampled lists with high probability? For instance, does $q=\Theta(\Delta^2)$ work in this case (perhaps for $\Delta$ sufficiently large or even $\omega(\log{n})$)? An alternative version of 
 this question, which to our knowledge is not studied even for PST, is to fix a choice of $q$ to be, say, $\Delta+1$, and instead determines what \emph{fraction} of vertices can be properly colored this way?

%------------------------------------------ Acknowledgement-------------------------------------------------
\section*{Acknowledgement} 
%\addcontentsline{toc}{section}{Acknowledgement}

We would like to thank Soheil Behnezhad, Yannic Maus, and Ronitt Rubinfeld for helpful discussions and their encouragement toward finding a simpler sublinear algorithms for $(\Delta+1)$ vertex coloring. 
We are also thankful to Chris Trevisan for pointing out a shorter proof of~\Cref{lem:warmup-RHS-suffices} included in this version of the paper. Sepehr Assadi is additionally grateful to Yu Chen and Sanjeev Khanna for their prior collaboration in~\cite{AssadiCK19} that formed the backbone of
this work. Finally, we are grateful to the anonymous reviewers of TheoretiCS for many invaluable comments and suggestions that helped tremendously with the presentation of the paper.  

\newpage

\printbibliography

%------------------------------------------ Warm-up-------------------------------------------------

% \input{warm-up}
\clearpage

\appendix
\section{Appendix: A Self-contained Sublinear Time Algorithm}\label{sec:warm-up}

In addition to our main result, we also present a very simple sublinear time algorithm for $(\Delta+1)$ coloring in the following theorem. This algorithm has recently and independently been discovered in~\cite{ferber2025improved} 
and has also appeared as part of lecture notes for different courses~\cite{Lec24,HW} at this point. Finally, this algorithm can also be seen as a simple adjustment to the sublinear time $(\Delta+o(\Delta))$ coloring algorithm of~\cite{MorrisS21}.

\begin{theorem}\label{thm:warm-up}
	There is a randomized algorithm that given any graph $G$ with maximum degree $\Delta$ via adjacency list and matrix access, 
	outputs a $(\Delta+1)$ coloring of $G$ in  $O(n \cdot \sqrt{n\log{n}})$ expected time. 
\end{theorem}

We note that unlike our sublinear time algorithm in~\Cref{thm:sublinear} which was \emph{non-adaptive}, namely made all its queries in advance before seeing the answer to them, the current algorithm is adaptive and needs to 
receive the answer to each query before deciding its next query. 

The algorithm in~\Cref{thm:warm-up} is quite similar to the standard greedy algorithm for $(\Delta+1)$ coloring. It iterates over the vertices
and colors each one greedily by finding a color not used among the neighbors of this vertex yet (which exists by pigeonhole principle). However, unlike the greedy algorithm, $(i)$ it crucially needs to iterate
over the vertices in a random order, and, $(ii)$ instead of 
iterating over the neighbors of the current vertex to find an available color, it samples a color randomly for this vertex and then iterates over all vertices with this color to make sure they are not neighbor to the current vertex. 
Formally, the algorithm is as follows. 

\begin{algorithm}\label{alg:warmup}
\SetKwFor{For}{For}{}{end for}
\caption{An (adaptive) sublinear time algorithm for $(\Delta+1)$ vertex coloring. }

Let $C_1,C_2,\ldots,C_{\Delta+1}$ be the \textbf{color classes} to be output at the end, initially set to empty. \;

Pick a permutation $\pi$ of vertices in $V$ uniformly at random.\;

\For{$v \in V$ in the order of the permutation $\pi$:}{

Sample $c \in [\Delta+1]$ uniformly at random. \label{line:reset} \;
For every vertex $u \in C_c$, check if $(u,v)$ is an edge in $G$; if \emph{Yes}, restart from Line~\eqref{line:reset}. \;
If the algorithm reaches this step, color the vertex $v$ with $c$ and add $v$ to $C_c$. \;
}

\end{algorithm}

It is easy to see that this algorithm never outputs a wrong coloring (namely if it ever terminates, its answer is always correct). Any new vertex colored does not create a conflict with previously colored vertices (given the algorithm explicitly checks to not color $v$ with a color $c$ if one of its neighbors is already colored $c$) and thus at the end, there cannot be any monochromatic edge in the graph. The interesting part of the analysis is to show that the algorithm terminates quickly enough, which is captured by the following lemma. 

\begin{lemma}\label{lem:warmup-runtime}
	For any input graph $G=(V,E)$, the expected runtime of~\Cref{alg:warmup} is 
	\[
	O\left(\frac{n^2}{\Delta} \cdot \log{\Delta}\right) ~\text{time}. 
	\] 
\end{lemma}

\noindent
To continue, we need to set up some notation.  For any vertex $v \in V$, we define the random variable $X_v$ as the number of $(u,v)$ queries checked by the algorithm in the for-loop of coloring $v$. Additionally, for any $v \in V$ and permutation $\pi$ picked over $V$, define $\Nback{v}$ as the neighbors $u$ of $v$ with $\pi(u) < \pi(v)$, namely the ones that are colored \emph{before} $v$ by Algorithm~\Cref{alg:warmup}. Let $\degback{v}:= \card{\Nback{v}}$. We start with the basic observation that the runtime of~\Cref{alg:warmup} can be stated in terms of the variables $\set{X_v}_{v \in V}$. 

\begin{observation}\label{obs:warmup-Xv}
	The expected runtime of~\Cref{alg:warmup} is $O(\sum_{v \in V} \expect{X_v})$. 
\end{observation}
\begin{proof}
	By definition, $X_v$ is the number of queries checked in the for-loop of coloring $v$. The algorithm repeats this for-loop for all $v \in V$, so the expected runtime is the number of all queries checked in this algorithm which is proportional to $\expect{\sum_{v\in V} X_v}$. Applying linearity of expectation concludes the proof. 
\end{proof}

Our task is now to bound each of $\expect{X_v}$ for $v \in V$ to bound the runtime of the algorithm using~\Cref{obs:warmup-Xv}. 

\begin{lemma}\label{lem:warmup-Xv-bound}
	For any vertex $v \in V$ and any choice of the permutation $\pi$:  
	\[
	\expect{X_v \mid \pi} \leq \frac{n}{\Delta+1-\degback{v}}.
	\]
\end{lemma}

To prove~\Cref{lem:warmup-Xv-bound}, we first need the following claim. 

\begin{claim}\label{clm:warmup-Xv-step}
	Fix any vertex $v \in V$, any choice of the permutation $\pi$, and any assignment of colors $C(u_1),C(u_2),\ldots$ by~\Cref{alg:warmup} to all vertices that appear before $v$ in $\pi$. 
	Then, 
	\[
	\expect{X_v\mid \pi} = \frac{\expect{\card{C_c}\mid \pi}}{\Pr\paren{\text{$c$ does not appear in $\Nback{v}$} \mid \pi, C(u_1),C(u_2),\ldots}},
	\]
	where in the RHS, both the expectation and the probability are taken with respect to a color $c$ chosen uniformly at random from $[\Delta+1]$. 
\end{claim}
\begin{proof}
	Define the colors $B(v)$ as the set of colors that appear in $\Nback{v}$, that is 
	\[
	B(v) := \set{c \in [\Delta+1] \mid  \text{there exists $u\in\Nback{v}$ with $c(u)=c$}}.
	\]
	
	For every color $c$, if $c$ is in $B(v)$ then $v$ cannot be colored by $c$, and otherwise it can. The probability of picking each color $c$ is $1/(\Delta +1)$. 
	For $c \notin B(v)$, the number of needed queries before coloring $v$ is $\card{C_c}$. 
	For $c \in B(v)$, the algorithm first needs to check up to $\card{C_c}$ queries to know this color is not available to $v$, and then it simply needs to repeat the same exact process.
	As such, 
	\[
	\expect{X_v\mid \pi}  = \sum_{c \notin B(v)} \frac{1}{\Delta+1}\card{C_{c}} + \sum_{c \in B(v)} \frac{1}{\Delta+1}(\card{C_{c}}+ \expect{X_v\mid \pi} ) = \expect{\card{C_c}\mid \pi} + \frac{\card{B(v)}}{\Delta +1} \cdot \expect{X_v\mid \pi}.
	\]
	
	We can also define the probability of $c$ not appearing in $\Nback{v}$ in terms of  $\card{B(v)}$ as below:
	
	\[
	\Pr\paren{\text{$c$ does not appear in $\Nback{v}$} \mid \pi, C(u_1),C(u_2),\ldots} = 1 - \frac{\card{B(v)}}{\Delta+1}. 
	\]
	
	By solving the recursive equation above, we get that
	\[
	\expect{X_v\mid \pi} =\frac{\expect{\card{C_c}\mid \pi}}{\Pr\paren{\text{$c$ does not appear in $\Nback{v}$} \mid \pi, C(u_1),C(u_2),\ldots}}. \qedhere
	\]
\end{proof}

Using~\Cref{clm:warmup-Xv-step}, we can conclude the proof of~\Cref{lem:warmup-Xv-bound}. 
\begin{proof}[Proof of~Lemma~\ref{lem:warmup-Xv-bound}]
	Each color $c' \in [\Delta+1]$ is chosen with probability $1/(\Delta+1)$ in Line~\eqref{line:reset} of the algorithm. Thus, 
	\[
	\expect{\card{C_c} \mid \pi} = \sum_{c'\in [\Delta +1]}{\Pr\paren{c=c' \mid \pi}\cdot \card{C_{c'}} } \leq \frac{n}{\Delta +1}, 
	\]
	as the sets $\set{C_{c'} \mid c' \in [\Delta+1]}$ are disjoint and partition the already-colored vertices which are at most $n$ vertices. 
	At this point, at most $\degback{v}$ colors have been used in the neighborhood of $v$ and thus cannot be used to color $v$. As such,
	\[
	\Pr\paren{\text{$c$ does not appear in $\Nback{v}$} \mid \pi, C(u_1),C(u_2),\ldots} \geq 1-\frac{\degback{v}}{\Delta +1}.
	\]
	Using~\Cref{clm:warmup-Xv-step} and the above bounds, we conclude
	\begin{align*}
		\expect{X_v \mid \pi } &= \frac{\expect{\card{C_c} \mid \pi}}{\Pr\paren{\text{$c$ does not appear in $\Nback{v}$} \mid \pi, C(u_1),C(u_2),\ldots}} \\
		&\leq \frac{n}{\Delta +1 -\degback{v}}.\qedhere
	\end{align*}
\end{proof}

By~\Cref{lem:warmup-Xv-bound} (and~\Cref{obs:warmup-Xv}), for any choice of the permutation $\pi$ in~\Cref{alg:warmup},  
\begin{align}
	\expect{\text{runtime of~\Cref{alg:warmup}} \mid \pi} = O(1) \cdot \sum_{v \in V} \frac{n}{\Delta+1-\degback{v}}.  \label{eq:warmup-suffices}
\end{align}

We now consider the randomness of  $\pi$ to bound the RHS above in expectation over $\pi$. 
\begin{lemma}\label{lem:warmup-RHS-suffices}
	We have, 
	\[
	\Exp_{\pi}\bracket{\sum_{v \in V} \frac{n}{\Delta+1-\degback{v}}} = O\left(\frac{n^2}{\Delta} \cdot \log{\Delta}\right).
	\] 
\end{lemma}
\begin{proof}
	Using the linearity of expectation, we have
	\[
	\Exp_{\pi}\bracket{\sum_{v \in V} \frac{n}{\Delta+1-\degback{v}}} =n \cdot \sum_{v\in V} \Exp_{\pi}\bracket{\frac{1}{\Delta+1-\degback{v}}}.
	\]
	
	For each permutation $\pi$, $\degback{v}$ depends on $v$'s relative position in the permutation $\pi$ with respect to its neighbors and it can vary from $0$ to $\deg(v)$. Each of these positions happens with the same probability ${1}/({\deg(v)+1})$. 
	Hence, for every vertex $v\in V$,
	\begin{align*}
		\Exp_{\pi}\bracket{\frac{1}{\Delta+1-\degback{v}}} &= \sum_{d=0}^{\deg(v)}\frac{1}{\deg(v)+1}\cdot \frac{1}{\Delta +1-d} \leq \sum_{d=0}^{\Delta}\frac{1}{\Delta+1}\cdot \frac{1}{\Delta +1-d} = O\left(\frac{\log{\Delta}}{\Delta}\right),
	\end{align*}
	where the inequality holds because if we let $A := \set{n/(\Delta + 1 - d)}_{d=0}^{\Delta}$, then, in the LHS, we are taking the average of the smallest $\deg(v) + 1$ numbers in $A$, whereas in the RHS we are taking the average of all of $A$.
	Plugging in this bound in the equation above concludes the proof. 
\end{proof}

\Cref{lem:warmup-runtime} now follows immediately from \Cref{eq:warmup-suffices} and \Cref{lem:warmup-RHS-suffices}. We can now use this to wrap up the proof of~\Cref{thm:warm-up}. 

\begin{proof}[Proof of~Theorem~\ref{thm:warm-up}]
	
	First, we consider the case $\Delta \geq \sqrt{n\log n}$.  In this case, we use \Cref{alg:warmup}, which relies on the adjacency matrix to access the input graph.
	By~\Cref{lem:warmup-runtime} we know that this algorithm has expected runtime $O(\frac{n^2}{\Delta} \cdot \log{\Delta})$ which is $O(n\sqrt{n\cdot \log(n)})$ by the lower bound on $\Delta$.

	If $\Delta < \sqrt{n\log n}$ we can use the standard deterministic greedy algorithm for vertex coloring which has linear runtime of $O(n\Delta)$. This algorithm uses an adjacency list to access the graph. 
	This is again $O(n\sqrt{n\log{n}})$ by the upper bound on $\Delta$ in this case, concluding the proof.
\end{proof}

\end{document}